\documentclass[letterpaper, 10 pt, journal, twoside]{IEEEtran}% Comment this line out if you need a4paper

\IEEEoverridecommandlockouts                              % This command is only needed if 
\usepackage{graphics} % for pdf, bitmapped graphics files
\usepackage{epsfig} % for postscript graphics files
\usepackage{times} % assumes new font selection scheme
\usepackage{amsmath} % assumes amsmath package installed

\usepackage{amssymb}  % assumes amsmath package installed
\usepackage{amsthm}
\usepackage{upgreek}
\usepackage{dsfont}
\usepackage{subfigure}
\usepackage{mathtools}
\usepackage{multirow}  % for `\multirow` macro
\usepackage{cite}
\usepackage{lipsum} %%% to write big equation in a single column
\usepackage{cuted}
\usepackage{romannum}

\newtheorem{theorem}{Theorem}

\newtheorem{corollary}{Corollary}

\title{\LARGE \bf
Influence Dynamics and Consensus in an Opinion-Neighborhood based Modified Vicsek-like Social Network}
\author{Narayani~Vedam and~Debasish~Ghose%
\thanks{N. Vedam is a PhD candidate and D. Ghose is a Professor at the Indian Institute of Science, Bangalore-560012, India. (e-mails: narayaniv+dghose@iisc.ac.in)}}

\begin{document}
\maketitle
\thispagestyle{empty}

%\nocite{*}

%%%%%%%%%%%%%%%%%%%%%%%%%%%%%%%%%%%%%%%%%%%%%%%%%%%%%%%%%%%%%%%%%%%%%%%%%%%%%%%%
\begin{abstract}
	We propose a modified Vicsek-like model to interpret influence dynamics and opinion formation in social networks. We work on the premise that opinions of members of a group may be considered to be analogous to the direction of motion of a particle in space. Similar to bounded-confidence models, interactions are based on closeness of opinions. The interactions are modeled by an adaptive network with dynamic node and tie weights. A mix of individuals - rigid and flexible - is assumed to constitute liberal and conservative groups. We analyze emergent group behaviors subject to different initial conditions, agent types, their densities and tolerances. The model accurately predicts the role of rigid agents in hampering consensus. Also, a few structural properties of the dynamic network, resulting as a consequence of the proposed model have been established.
\end{abstract}

%
%
%%%%%%%%%%%%%%%%%%%%%%%%%%%%%%%%%%%%%%%%%%%%%%%%%%%%%%%%%%%%%%%%%%%%%%%%%%%%%%%%%
\section{INTRODUCTION}
\IEEEPARstart{H}{uman} opinions are moulded through a continuous process, subject to a host of influences. In-person meetings, visual and print media, etc., remain the chief exogenous factors. With the wide spread of the internet, there is a drastic change in the way opinions are formed. They are increasingly shaped by information diffused through tweets and posts on virtual platforms. Such spreads tend to influence the behavior of individuals and thereby the societies. Due to their relevance to social and economic problems, there is a growing interest to understand the underlying dynamics. 

Influence dynamics and formation of opinions have been extensively studied from the perspective of analytical modeling and psychology. Early 1900s have witnessed social experiments that have enormously helped in understanding human behavior. With their aid, in the following decades, several behavioral models were proposed. Among them, one of the earliest formulations \cite{french1956formal}, interprets the formation of groups. \cite{stone1961opinion} proposed a continuous-time model for joint decision making. DeGroot \cite{degroot1974reaching} came up with a discrete-time repeated averaging model. Similar models which investigate the impact of agent weights on consensus can be found in \cite{chatterjee1977towards,cohen1986approaching}. However, these models study agreement in groups. The DeGroot's model was extended \cite{friedkin1990social} to consider possible disagreements. In all these models, opinions are real numbers within a range. A more simplistic model is the voters model, where opinions are discrete. It was first proposed in \cite{clifford1973model}, where agents adopt opinions through random interactions. Variations to the model have been proposed over the years \cite{follmer1974random,orlean1995bayesian,latane1997self,weisbuch1999dynamical}. All the network models \cite{french1956formal,stone1961opinion,degroot1974reaching,chatterjee1977towards,cohen1986approaching,friedkin1990social,clifford1973model,follmer1974random,orlean1995bayesian,latane1997self,weisbuch1999dynamical}, although different in their level of detail, are linear, collaborative, and mostly rely on chanced interactions. One of the initial attempts with antagonistic interactions is \cite{altafini2013consensus}, where negative edge weights reflect the lack of synergy. Extensions to this work appear in \cite{altafini2015predictable,proskurnikov2016opinion}. In all these models, the agents are similar and their interactions are uninhibited - an agent can interact with any other agent.

The first nonlinear opinion model is by Krause \cite{krause1997soziale}. Ever since, similar models have been proposed \cite{deffuant2000mixing,sznajd2000opinion,hegselmann2002opinion,lorenz2007continuous,mor2011opinion,yang2014opinion,liu2013opinion}. They recommend repeated opinion averaging, but differ in their communication regimes. Among them, the bounded-confidence model \cite{hegselmann2002opinion} is well known; an agent engages with others who conform to its belief. The interactions among agents may change with time, thereby capturing an aspect of relationship dynamics. This addressed one of the earlier shortcomings, while overlooking the possibility of agent types. 

The impact of agent types on consensus has been widely studied. One of the earliest works \cite{weisbuch2002meet}, studies agents with different interaction thresholds. \cite{kou2012multi,fu2015opinion,liang2013opinion}, classify agents as closed, open and moderate, and analyse different group compositions. \cite{faure2002dynamics,galam2004contrarian,amblard2004role,acemoglu2010spread,acemouglu2013opinion,delellis2017steering} classify agents based on opinions, and discuss the effect of contrarians, extremists and forceful agents. \cite{moussaid2013social} analyses the impact of well-informed minorities in the midst of an uninformed majority. In \cite{porfiri2007decline,yildiz2011discrete}, extremist minorities in a fairly stubborn society, and stubborn agents in a voters model have been explored. \cite{ghaderi2014opinion} studies stubbornness of agents as a criterion in a cost function.

Similar to bounded-confidence models, in the context of non-equilibrium systems, Vicsek \textit{et al.} \cite{vicsek1995novel} proposed a model with biologically inspired interactions to study the motion of self-propelled particles. Accordingly, at each time-step a particle driven with a constant absolute velocity adopts the the average heading of those in its $r-$neighborhood with some added perturbation. This nearest neighbor rule has indicated consensus about the heading, despite the absence of centralized co-ordination and a dynamic neighborhood. This, coupled with the model's simplicity, makes it a popular choice for studying robotic swarms. It has been extensively used for data fusion in sensor networks \cite{ogren2004cooperative}, collaboration of UAVs \cite{stipanovic2004decentralized}, and in explaining the behavior of animal groups \cite{chazelle2009convergence,shaw1962schooling}.

In this work, we adapt the model in \cite{vicsek1995novel} since the nearest-neighbor rule with dynamic local interactions epitomizes influence spread in networks. We assume that beliefs of members of a social group are analogous to the direction of motion of a particle in space. Unlike all bounded-confidence models, we distinguish agents within a neighborhood based on familiarity. Heterogeneous agents with non-uniform interaction thresholds have been considered. The opinion distribution is modeled as a truncated Gaussian, in contrast with a uniform spread encountered in most of the existing works. This is consistent with the assumptions and experimental validation of opinions collected from social networks \cite{boccara2008models,de2014learning,chacoma2015opinion}.

It is evident that beliefs and the underlying interaction network co-evolve. In most of the existing works, the changes in network are either due to randomly induced ties or as a consequence of the bounded-confidence assumption. In addition, \cite{zimmermann2005cooperation} employs a reward-or-penalize strategy for evolution of ties.  The impact of rewiring is analysed in \cite{kozma2008consensus,nardini2008s}. In the context of complex systems, \cite{delellis2017steering}, \cite{boccaletti2006complex,delellis2010evolution} explore a dynamic member-set, network based on proximity rule, and edge snapping, respectively. Instead, here, a network derived from individual opinions, thresholds and interactions evolve in accordance with prescribed rules, reflecting dynamic interpersonal relationships. Besides, neither are all agents nor are all their interactions equal; this is not only modeled by dynamic weights of nodes and edges, but also by familiarity-based interactions. These aspects have been rarely explored in conjunction with opinion dynamics. While such a model may not include all the complexities of a social network, our studies show that interesting phenomena can be observed in the results obtained from this model which give some insights into the evolution of opinion and consensus in social networks.

\section{THE MODELS}
\subsection{The original Vicsek Model}
Consider $N$ particles that are restricted to move in a periodic square box. The particles are randomly placed in the box and have the same absolute velocity. Their initial headings ($\theta_{i}(0)$) are uniformly chosen at random from within a range $(-\pi,\pi]$. At every time-step, the particles assume the average heading ($\langle\theta(t)\rangle_{r}$) of those within a circle of radius $r$ surrounding them, with some added noise ($\Delta\theta$). This average direction is given by $\arctan[\langle\sin(\theta(t))\rangle_{r}/\langle\cos(\theta(t))\rangle_{r}]$.
\subsection{The Modified Vicsek-like Model and Corresponding Social Network} 
Consider $N$ agents, each with a belief $\theta_{i}(t) \in [0,\pi]$. An agent engages with those others, whose beliefs do not deviate from its own by more than a fixed tolerance ($\theta_{T_{i}}$). This is similar to interactions over virtual networks that are oblivious to geographic locations of agents and the distances separating them. The like-minded individuals who influence an agent, constitute its opinion neighborhood,
\begin{equation}
N_{i}(t)~:=~\{~j~:~|\theta_{i}(t)-\theta_{j}(t)|~\leq ~\theta_{T_{i}}\}.
\end{equation}
In Fig. \ref{fig1}, the vectors represent opinions and shaded sectors indicate tolerances. It is clear that the sectors of agents $1$ and $2$ overlap the opinion vectors of $2$ and $1$, respectively. The same can be observed with $2$ and $3$, indicating they are neighbors.
\begin{figure}[t!]	
	\centering
	\includegraphics[width=4.5cm]{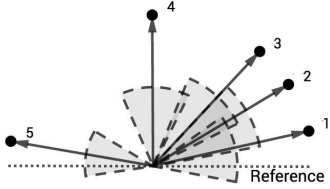}
	\vspace{-5pt}
	\caption{Opinion vectors with shaded sectors indicating tolerance ($\theta_{T_{i}}$)}	
	\label{fig1}\vspace{-15pt}
\end{figure}

The agents and their influence can be abstracted using vertices and directed edges of a graph $G(t)= (V,E(t))$. Its vertex set $V$ is invariable and is a collection of $N$ agents. Its edge set $E(t)$ is a collection of all the interactions at time $t$,
\begin{equation}
 E(t)~\subseteq~\{~(i,j)~:~j\in~N_{i}(t)~,~\forall i,~i~\neq~j~\}.
 \end{equation} 
 The directed edge $(i,j)$ is the influence of $i$ on $j$. To illustrate, consider the agents and their tolerance as portrayed in Fig.\ref{fig1}. The corresponding vertex and edge sets are $V=\{1,2,3,4,5\}$ and $E(t)=\{(1,2),(2,1),(2,3),(3,1)\}$, respectively. This is represented in Fig. \ref{fig2a} where there are bidirectional edges between vertex pairs $(1, 2)$ and $(2, 3)$. The beliefs of $4$ and $5$ are different from the rest, and are isolated. Fig. \ref{fig2b} depicts a more realistic scenario, where influences are not reciprocated. It can be observed that 2 influences 1, while 1 does not influence 2. Similarly, 3 influences 2 and not vice versa. 
\begin{figure}[t!]
	\centering
	\subfigure[Agents and their influence]{\includegraphics[width=4.cm]{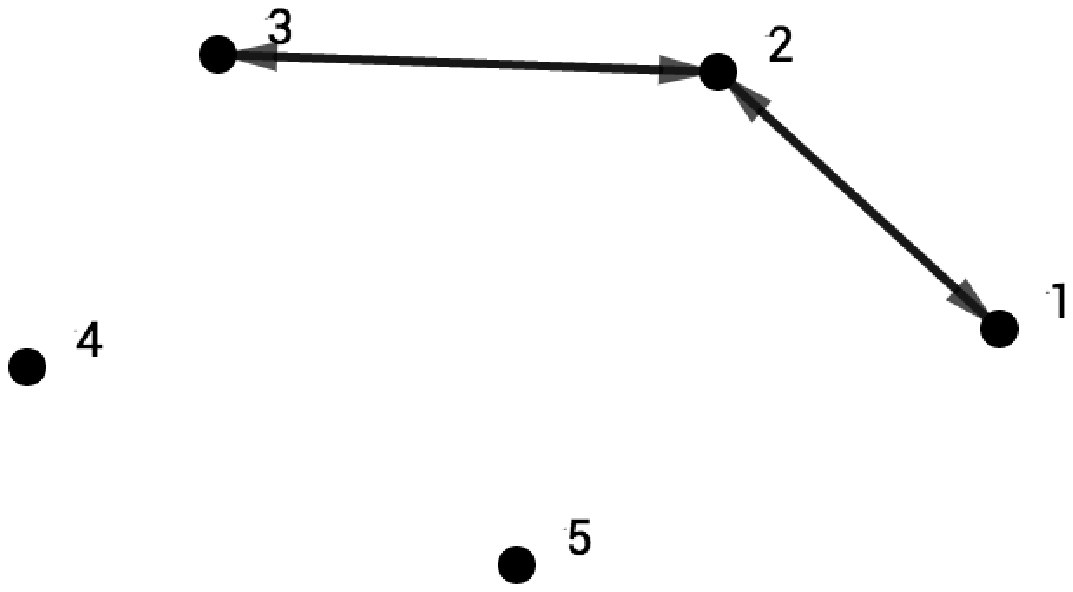}
		\label{fig2a}}
	\subfigure[A more realistic scenario]{\includegraphics[width=4.cm]{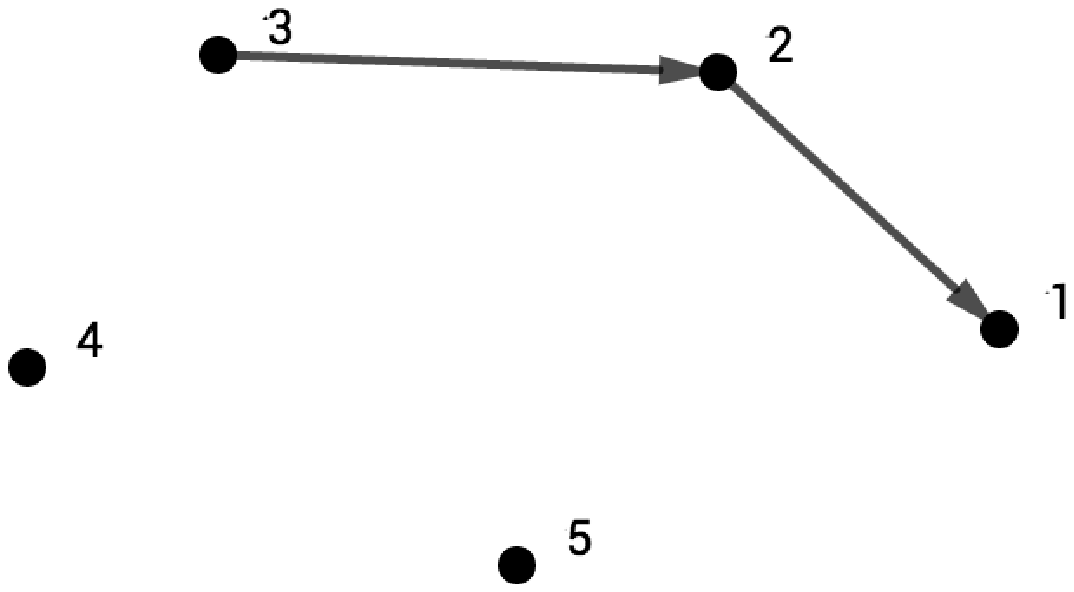}\label{fig2b}}
	\caption{Representation of agents and their interactions using a graph $G(t)$}
	\label{fig2}
	\vspace{-15pt}
\end{figure}

The ad-hoc nature of human interactions is modeled by asynchronous updates. When an agent modifies its belief under the influence of its neighbors ($\langle \theta_{i}(t)\rangle$), it is governed by 
\begin{align}\label{e1}
\begin{split}
&\theta_{i}(t+1) = \langle \theta_{i}(t)\rangle,\\ 
& \tan(\langle \theta_{i}(t)\rangle)=\frac{\sum\limits_{j\in N_{D_{i}}(t)\cup \bar{N}_{ND_{i}}(t)} w_{ij}(t) \sin(\theta_{j}(t))}{ \sum\limits_{j\in N_{D_{i}}(t)\cup \bar{N}_{ND_{i}}(t)}w_{ij}(t) \cos(\theta_{j}(t))}.
 \end{split}
 \end{align}
Whenever agents interact and update their beliefs, the network topology may change; the updated belief of an agent may impact its neighborhood set, and thereby, the network's edges.

The edges are classified as direct and non-direct ties. Direct ties are an agent's frequent contacts belonging to $N_{D_{i}}(t)$, and are its one-hop neighbors. Non-direct ties are occasional contacts belonging to $N_{ND_i}(t)$, and are an agent's two-hop or three-hop neighbors, such that,
\begin{equation}
N_{i}(t) = N_{D_{i}}(t)\cup N_{ND_{i}}(t).
\end{equation} 
This is essential because, like-minded agents may be unaware of each others' presence in a vast network. Empirical analyses establishing the importance of such influences can be found in \cite{aral2009distinguishing}. When an agent updates its belief using (\ref{e1}), it considers $N_{D_{i}}(t)$ and a randomly chosen subset, $\bar{N}_{ND_{i}}(t) \subseteq N_{ND_{i}}(t)$. This way, noise is inherent in our model; greater an agent's tolerance, more susceptible it is to such chanced interactions. However, there exist noise-free extensions to the Vicsek model that analyse heterogenous thresholds \cite{yang2006consensus}. 

Trust scores - $w_{ij}(t),~\text{with}~j~\neq~i$, in (\ref{e1}) associated with edges signify the strength of an influence. Initially, direct ties are assigned random trust scores ($w_{ij}(t)$). Between a pair of vertices $(j,i)$, a sequence of directed edges connecting them is a directed path. To compute the weight of a non-direct tie, all such paths between the vertex pair are identified. The weight of each path is defined as the product of weights of the edges that constitute them. The score of a temporary edge is the average of the weights of all such paths between a pair $(j,i)$;
\begin{align}
\begin{split}
&w_{ij}(t)=\frac{\sum\limits_{p=1}^{p_{max}}\prod\limits_{(m,n)\in E_{path}}w_{mn}(t)}{p_{max}},\\
&E_{path}=\{(m,n)|\ (m,n)\in \text{the directed path}~ (j,i) \},\\
&p_{max}=|E_{path}|,~\sum_{j,j\neq i}w_{ij}(t)~=~ 1-w_{ii}(t).
\end{split}
\end{align}
Hence, the influences of direct ties are stronger. Evidently, the trust scores are proposed based on familiarity of agents - frequent or occasional contact. This is unlike \cite{da2015sudden} which is an exclusive treatment of trust based on several psychological parameters. It is also different from \cite{fu2008reputation,delellis2017evolving}, where trust or reputation is built over time through interactions.

The self-weights ($w_{ii}(t)$) reflect an agent's status in the network. To measure them, we use Katz's centrality \cite{katz1953new}, 
\begin{equation}
W(t)~=~\beta(I-\alpha A^{T}(t))^{-1}~.~\mathbf{1},
\end{equation} 
where, $\beta$ is a positive non-zero bias, $\alpha <1/\lambda_{max}$, $A(t)$ is the adjacency matrix of the underlying network, and $W(t)$ is a vector of self-weights. Accordingly, the importance of a node is based not only on its degree, but also on the nature of its contacts. Essentially, a node's connection to influential nodes yields a higher self-weight than ties with less prominent ones. Additionally, a socially important individual is less obligated to yield to influence, which is reflected in (\ref{e1}), where, self-weights quantify the importance accorded to one's opinion. 

Upon interactions, new connections are forged only with non-direct ties, since, humans do not reach-out beyond them \cite{christakis2009connected}. A new tie is gained when the cumulative number of random interactions between agent $i$ and its non-direct tie $j$, exceeds its sociability index ($s_{i}$). Under the assumption that socially important agents do not forge ties easily \cite{bruggeman2013social}, the index is assumed to be proportional to self-weight;
\begin{equation}
s_{i} = K_{C}w_{ii},
\end{equation}
 where, $K_{C}$ is a proportionality constant, chosen based on the network size. Also, it has been observed that inter-personal bonds are not easily lost \cite{leary2017need}. Therefore, only the ties between agents with dissimilar opinions disappear; agent $i$ loses a direct tie $j$ when the following condition is violated, 
 \begin{equation}
 |\theta_{i}(t)-\theta_{j}(t)|\leq\theta_{T_{i}}.
 \end{equation}

\section{PROBLEM FORMULATION}
Consider $N$ agents, each with a belief $\theta_{i}(t)$ and a tolerance $ \theta_{T_{i}}$. The beliefs are updated using (\ref{e1}), and the interactions are governed by rules prescribed in the previous section. As beliefs change, the underlying network topology evolves. This interplay yields different collective behaviors. We are interested in studying consensus or the lack of it, subject to different (\romannum{1}) initial opinion spreads, (\romannum{2}) constituent agent types, (\romannum{3}) individual tolerance, and (\romannum{4}) densities of agents.

\subsection{Classification of Agents}
The agents are grouped based on their tolerances;

\subsubsection{Rigid} These are stubborn individuals that dialogue with those with similar beliefs. Their intolerance to diverse beliefs is modeled with lower thresholds, $\theta_{T_{i}}\in[0,\theta_{R}]$.

\subsubsection{Flexible} They admit beliefs very different from their own, and are modeled as agents with higher tolerances, $\theta_{T_{i}}\in[\theta_{F_{1}},\theta_{F_{2}}]$, such that, $\theta_{F_{2}}>\theta_{F_{1}}>\theta_{R}$.

\cite{kou2012multi} - \cite{liang2013opinion} employ a similar strategy, where, agents are classified based on predetermined values of threshold.
\subsection{Classification of Groups}

The groups are classified depending on the initial distribution of their agents' beliefs. In real, opinions may not be uniformly spread.  Instead, there could be a popular opinion with a few others scattered in the vicinity. This pattern is modelled using a truncated Gaussian model, whose mean ($\mu$) and standard deviation ($\sigma$), represent the popular opinion and the nature of the group, respectively. This is not only consistent with \cite{boccara2008models}, but are also subtantiated by experiments carried out in \cite{de2014learning} and \cite{chacoma2015opinion}. Accordingly, the groups are classified as,
\subsubsection{Conservative} The agents' initial beliefs are typically in the vicinity of a popular opinion. Using a truncated Gaussian distribution to model the initial spread, opinions take values around the mean ($\mu$); since, such groups allow a few agents with contrary opinions (in Fig. \ref{fig5b}). This is addressed by choosing a small spread, $\sigma \in [0,\theta_{C}]$.
\subsubsection{Liberal} It admits varied opinions, and this diversity is reflected in a larger spread, $\sigma \in [\theta_{L_{1}},\theta_{L_{2}}]$, where, $\theta_{L_{2}}>\theta_{L_{1}}>\theta_{C}$. Unlike a conservative group, there are fewer agents with beliefs close to the mean (in Fig. \ref{fig5d}).
\begin{figure}[t!]
	\subfigure[PDF of a conservative group]{\includegraphics[width=4.cm]{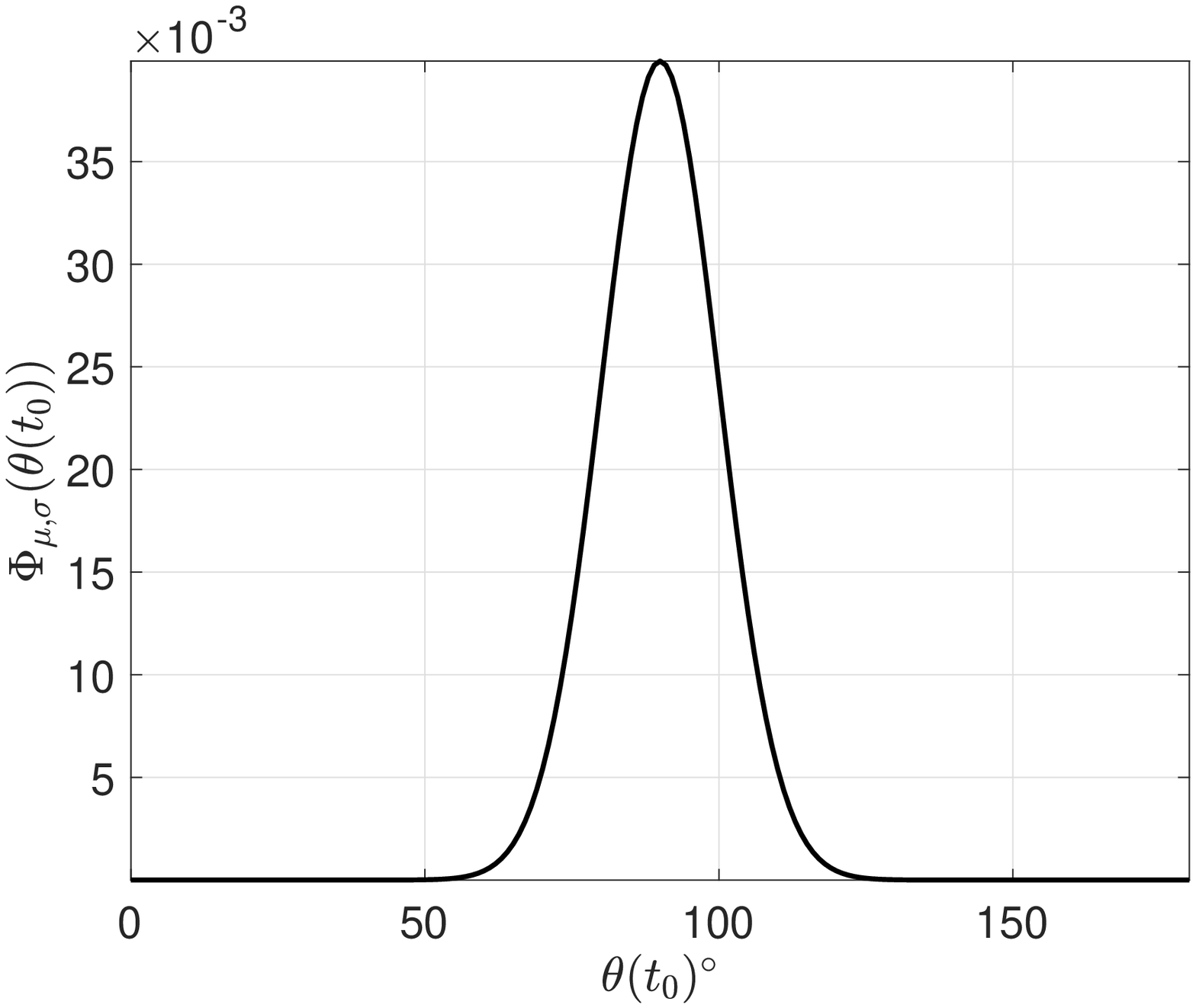}
	\label{fig5a}}
	\hfil
	\subfigure[Conservative group opinions]{\includegraphics[width=4.cm]{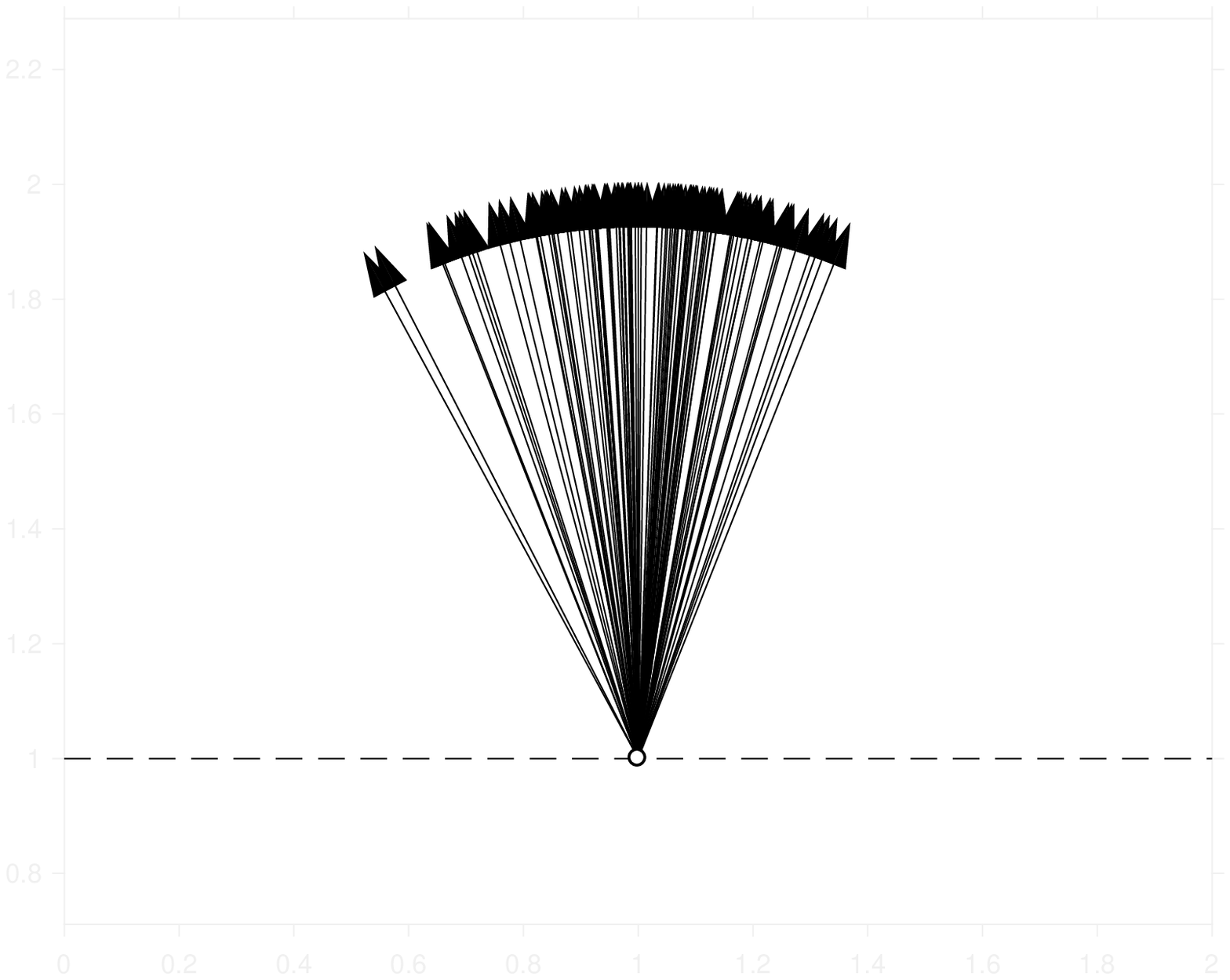}\label{fig5b}}
		\subfigure[PDF of a liberal group]{\includegraphics[width=4.cm]{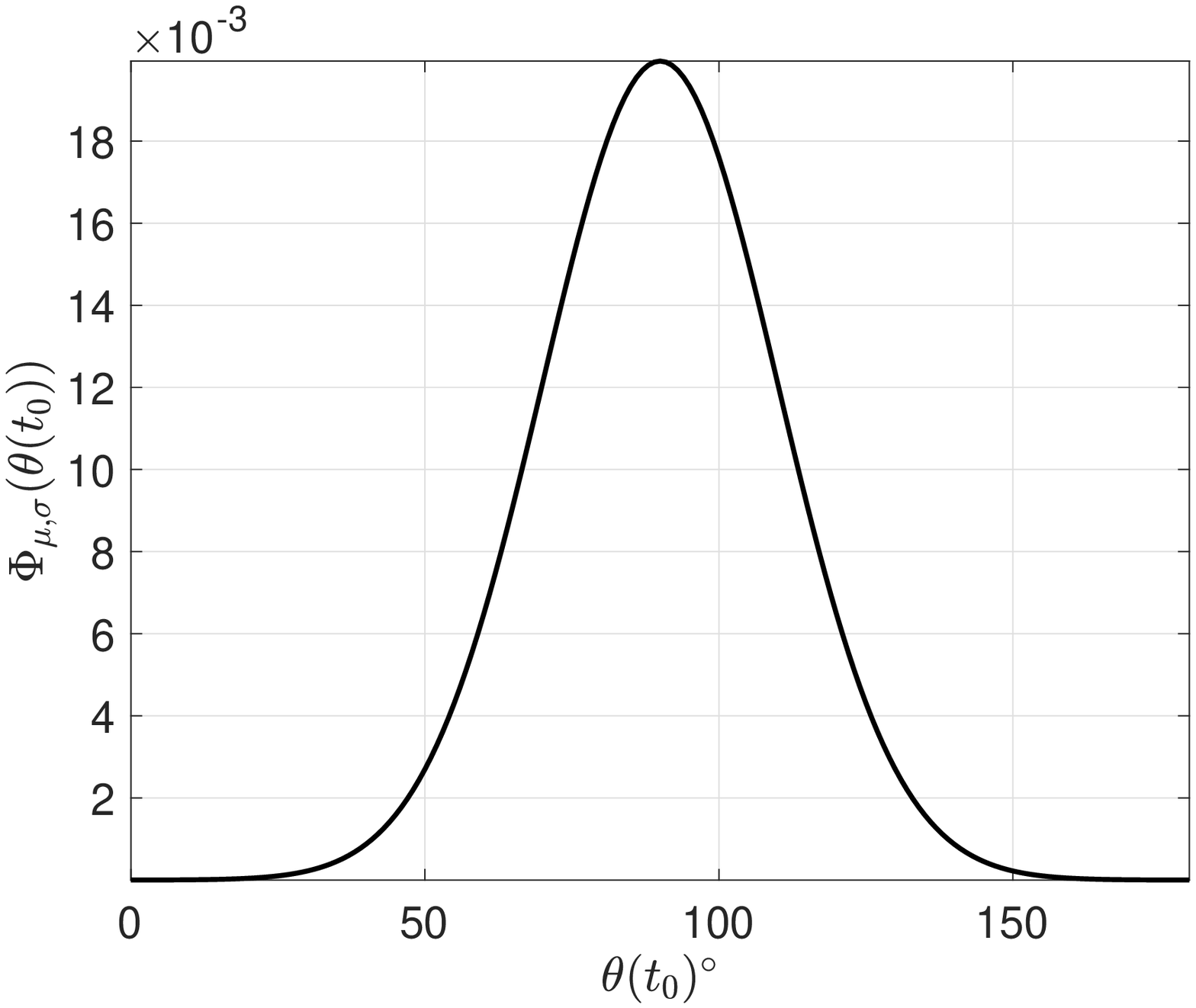}\label{fig5c}}
	\hfil
	\subfigure[Liberal group opinions]{\includegraphics[width=4.cm]{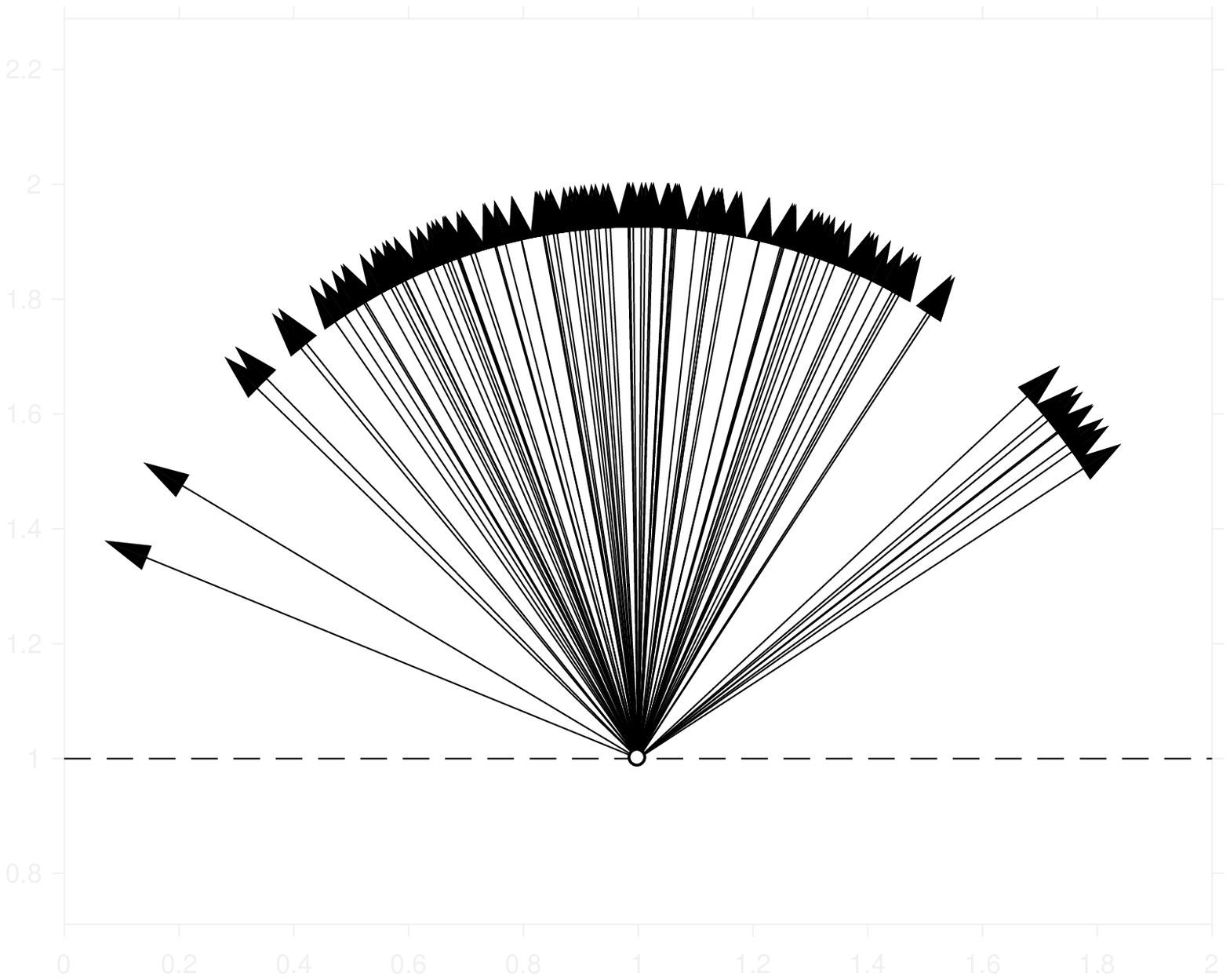}\label{fig5d}}
	\caption{Truncated Gaussian distribution of initial opinions in groups}
	\label{fig5}\vspace{-15pt}
\end{figure}

The names accorded to these groups have no political connotations, and have been chosen as they best fit the description.
 
\section{STRUCTURAL PROPERTIES OF THE NETWORK}

Consensus in a group can either be a global or a local phenomenon. Although global consensus is ideal, in real, it often happens that the members of a group separate to form smaller opinion groups. It would be interesting to discern the possible causes of this phenomenon. To this end, we undertake a diagnostic approach to identify the edges and the corresponding agents responsible. As a consequence of the proposed model, the collection of these edges have been observed to possess some unique characteristics.

Consider an initial network, $G(t_{0})=(V,E(t_{0}))$. The beliefs updated using (\ref {e1}) is said to have converged when,
\begin{equation}
\theta({t_{k}})-\theta({t_{k}-1})~=~0,~\text{for some}~k>0.
\end{equation}
The converged network is $G(t_{f})$, where, $t_{f}:=t_{k}$. When there is local consensus, there will be more than one connected component in $G(t_{f})$. To identify the responsible edges, $G(t_{f})$ is compared with $G(t)$. The edge cut-set ($E_{cs}(t)$) is a collection of such edges with interesting properties: (\romannum{1}) There exists a unique set of agents and edges that are responsible for group splits, (\romannum{2}) The newly formed ties that are acquired as the network evolves, are not responsible for the factions formed. These properties are established in the following theorems.
\begin{theorem}
	The edge cut-set $E_{cs}(t)$ of a $(G(t), G(t_{f} ))$ pair is unique and minimal.
\end{theorem}

\begin{proof}
	Assume $E_{cs}(t)$ is not unique, then,
	\begin{equation}
	\exists\ E_{cs}^{'}(t)\subset\{ E(t)\setminus E_{cs}(t)\} \ni  (G(t)-E_{cs}^{'}(t))=G(t_{f})
	\end{equation}
	This implies that there exist edges in $G(t)$ between components of $G(t_{f})$ that are not in $E_{cs}(t)$. But, $E_{cs}(t)$ is the set of all such cross-connections, and therefore $E_{cs}^{'}(t)=E_{cs}(t)$. This proves uniqueness of $E_{cs}(t)$.
	
	Now, Assume that $E_{cs}(t)$ is not minimal, then,
	\begin{equation}
	\exists (i,j)\in E_{cs}(t)\ni (G(t)-\{E_{cs}(t)\setminus(i,j)\})=G(t_{f})
	\label{e13}
	\end{equation}
	But, $(i,j)\in E_{cs}(t)$, implies that it is an edge in $G(t)$ between two groups of $G(t_{f})$, and therefore $(i,j)$ cannot be redundant or (\ref{e13}) cannot be true. This proves minimality. 
\end{proof}

\begin{theorem}
	If $E_{cs}(t_{0})$ is an edge cut-set corresponding to a pair $(G(t_{0}),G(t_{f}))$, then any $E_{cs}(t)$ corresponding to a pair $(G(t),G(t_{f})),\  t_{0}< t \leq t_{f}$, is a sub-set of $E_{cs}(t_{0})$. 
\end{theorem}
\begin{proof}
	Let $E_{cs}(t_{0})$ be the edge cut-set corresponding to $(G(t_{0}),G(t_{f}))$. For any $t>t_{0}$, we know that $E_{cs}(t)$ will have no new edges between vertices in different components of $G(t)$. Thus, it is apparent that no new edge is added to $E_{cs}(t),\forall t>t_{0}$, and therefore $E_{cs}(t)\subseteq E_{cs}(t_{0})$.
\end{proof}

\begin{corollary}
	For any $(G(t_{0}),G(t_{f}))$ pair, the edge cut-set $E_{cs}(t_{0})$ is the largest minimal set among $\{E_{cs}(t),\ \forall t>t_{0}\}$.
\end{corollary}

\section{SIMULATION RESULTS \& DISCUSSIONS}

A group of 100 agents has been considered with the simulation time set at 250 units. Unlike in \cite{fu2015opinion,liang2013opinion}, where, tolerance of agents is assigned according to some distribution, here, all the agents of a type are assumed to have a fixed tolerance, consistent with \cite{kou2012multi}. For instance, all rigid agents have a tolerance $\theta_{R} \in \{10^\circ,30^\circ\}$, while all flexible agents have a tolerance $\theta_{F} \in \{40^\circ,80^\circ\}$. The distribution of initial opinions for conservative and liberal groups are set at $\mu = 90^\circ$ with $\sigma \in \{10^\circ,15^\circ\}$ and $\sigma \in \{20^\circ,25^\circ\}$, respectively. The initial network is derived from opinions and tolerances. There are no assumptions made regarding the structure of the initial network. To prevent trivial initial networks - like well-connected ones - the out-degree of each vertex is capped at 25. The constant of proportionality, $K_{C}$, is set to 100.

To account for diverse opinions, network topologies, and interaction patterns, Monte Carlo simulations have been carried out. To this end, we define four group configurations,

\begin{enumerate}
	\item Configuration 1 is a conservative group with all rigid agents, or more rigid agents than flexible ones.
	\item Configuration 2 is a conservative group with all flexible agents, or more flexible agents than rigid ones.
	\item Configuration 3 is a liberal group, where, the density of agents is similar to Configuration 1.
	\item Configuration 4 is also a liberal group, where, the density of agents is similar to Configuration 2.
	
\end{enumerate}
 Each configuration with fixed agent tolerances is subjected to 100 simulation runs. Each run corresponds to a particular opinion distribution, fixed agents' density and tolerance - $\{10^\circ, 40^\circ\},~\{10^\circ, 80^\circ\},~\{30^\circ, 40^\circ\}~\text{or}~\{30^\circ, 80^\circ\}$. Starting with a group of all flexible agents, we study the impact of density of rigid agents on consensus. This is quantified by, 
\begin{align}
\small
\begin{split}
&\text{Rate of Consensus} = \frac{\text{ No. of times group consensus is achieved}}{\text{No. of simulation trials}}.
\end{split}
\end{align}
 The consensus rate for a configuration is essentially the average occurence of group consensus.

\subsection{Conservative Groups}

In such groups, when flexible members are in a majority, the dialogue between agents mostly yields consensus. This is observed in Fig. \ref{fig6a}, where the rate of consensus is high. On increasing the number of rigid members in the group, factions are formed. This is captured by a marginal drop in the consensus rate in Fig. \ref{fig6a}. This is because, the rigid members of the group could lose ties with their neighbors upon marginally changing their beliefs. Such changes to the neighborhood results in an eventual split of an originally connected group, rendering it highly intolerant to unpopular opinions. Additionally, from the plots in Fig. \ref{fig6a}, it can be inferred that flexible agents facilitate group consensus. This is substantiated by a drastic decline in consensus rate when their population is barely 10\%. Another key observation is regarding the tolerance of agents; as the tolerance of rigid agents is increased from $10^\circ$ to $30^\circ$, there is a remarkable improvement in the rate of consensus.

\begin{figure*}[t!]
	\centering
	\subfigure[Conservative group with $\sigma = 10^{\circ}$ (top) and $\sigma = 15^{\circ}$ (bottom) ]{\includegraphics[width=0.45\textwidth]{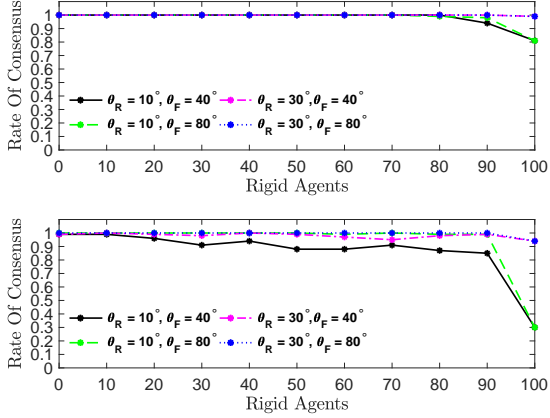}
		\label{fig6a}\vspace{-10pt}}
	\hfil
	\subfigure[Liberal group with $\sigma = 20^{\circ}$ (top) and $\sigma = 25^{\circ}$ (bottom)]{\includegraphics[width=0.45\textwidth]{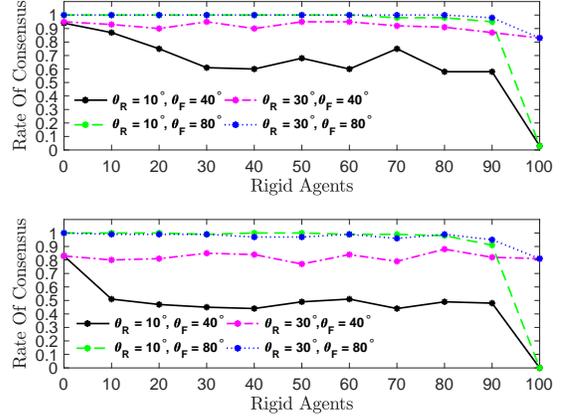}\label{fig6b}\vspace{-10pt}}
		\caption{Rate of consensus in conservative and liberal groups with varying agent densities and individual tolerance}
	\label{fig6}\vspace{-15pt}
\end{figure*}

\subsection{Liberal Groups}	

In these groups, the rate of consenus is impacted not only by the presence of rigid agents, but also a wider belief distribution. Intuitively, a group with all flexible individuals should have a relatively high rate of consensus for different values of $\theta_{T_{i}}$ and different initial spreads, which is also observed in Fig. \ref{fig6}. However, because of the belief spread, unlike conservative groups, the behaviour of a liberal group is sensitive to the inclusion of even a small number of rigid agents. Upon the inclusion of rigid agents in an all flexible group, there is a significant drop in the rate of consensus, which can be observed in Fig. \ref{fig6b}. The earlier observation about flexible agents being facilitators of group consensus holds. This is again justified by an eventual decline in consensus rate when their population is reduced to $10\%$ or lower. Similarly, an increase in tolerance of rigid agents positively impacts consensus. 

\begin{figure}[t!]
	\centering
	\subfigure[Group of rigid agents ($\mu = 90^\circ,~\sigma = 20^{\circ}$) ]{\includegraphics[width=5.5cm]{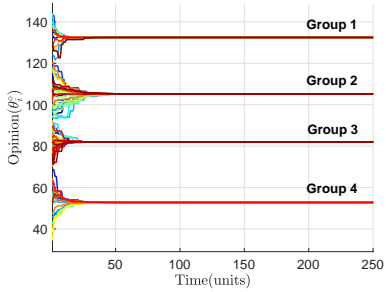}
		\label{fig8e}\vspace{-10pt}}
	
	\subfigure[$t=0$ ]{\includegraphics[width=4.cm]{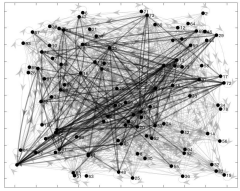}
		\label{fig8a}}
	\hfil
	\subfigure[$t=5$]{\includegraphics[width=4.cm]{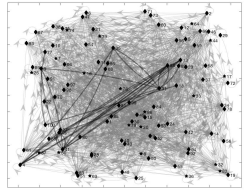}\label{fig8b}}
	
	\subfigure[$t=10$]{\includegraphics[width=4.cm]{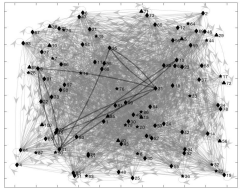}
		\label{fig8c}}
	\hfil
	\subfigure[$t=15$]{\includegraphics[width=4.cm]{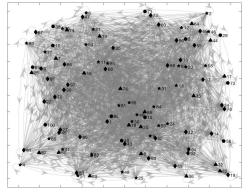}\label{fig8d}}
	\caption{Evolution of opinions ((a)) and the underlying network ((b)-(e)) indicating $E_{cs}(t)$ (bold edges) and factions (different node markers)}
	\label{fig8}\vspace{-15pt}
\end{figure}

Whenever there is a decline in the rate of consensus in any of the four configurations (see Fig. 4), it is indicative of instances when factions are formed. Higher the number of such instances, lower is the rate of consensus. In Fig. \ref{fig8e}, the evolution of opinions in a group of rigid individuals is illustrated. The lack of consensus results in more than one opinion group, which are formed when individuals with dissimilar opinions lose their ties belonging to the edge cut-set ($E_{cs}(t_{0})$). From Fig. \ref{fig8a} - \ref{fig8d}, it can be seen that the number of ties in the set decreases over time, $E_{cs}(t_{15})\subset E_{cs}(t_{10})\subset E_{cs}(t_{5}) \subset E_{cs}(t_{0})$. It can be observed that the new ties formed do not cause group splits. The ties in $E_{cs}(t_{0})$ are all lost when the network converges (Fig. \ref{fig8d}), resulting in four groups.

Thus, a conservative group with flexible agents in majority evolves into a moderately tolerant society; where minor disagreements do not create break-away factions. However, with the inclusion of rigid agents, despite the narrow initial spread of beliefs, there may be opinion groups formed. Essentially, a conservative society is less accommodative of rebelling minorities, which is similar to the observation in \cite{porfiri2007decline}. On the contrary, a liberal group with rigid agents evolves into a pluralistic society, while one with all flexible individuals always tends to form agreeable groups. Although not with similar societal construct or comparable adaptive interactions, the works \cite{kou2012multi,fu2015opinion,liang2013opinion} report similar predictions. Further, the role played by flexible agents as prime facilitators of group consensus, is worth reiterating. In-line with \cite{kou2012multi,fu2015opinion,liang2013opinion}, \cite{yang2006consensus}, an increase in the individual tolerance has shown to significantly transform the outlook of a group into a more progressive one.

\section{CONCLUSIONS}
We discussed a modified version of the Vicsek model to study influence dynamics, its spread and consequence on opinions. The heading angles of agents in \cite{vicsek1995novel} are considered analogous to beliefs, while ignoring their physical distances. Like in bounded-confidence models, a closest-opinion rule has been proposed. In comparison with \cite{kou2012multi,fu2015opinion,liang2013opinion}, the proposed model is equipped to handle the subtleties in human interactions, and opinion formation. 

The structural properties of the evolving network that are a consequence of the proposed model have been analytically established. With these preliminary results, in subsequent works, we intend to develop them further to predict agents responsible for group splits and regulate group behavior. To emulate real-life scenarios, we have considered different agent types, groups, and their compositions. In addition to confirming anticipated behavior, the results are in sync with similar reports in literature. 

Overall, the modified Vicsek-like model seems to demonstrate the potential to explain many observed phenomena in influence dynamics, group fragmentation and the role of agents. To realise its maximum potential, several avenues and future research directions can be pursued. The effect of opinion density on agreement values can be explored. In addition, the impact of leaders on group behavior may be evaluated.

\bibliographystyle{IEEEtran}
\bibliography{ref1}

\end{document}